\documentclass{llncs}
\usepackage{amsmath,amssymb,amsfonts,stmaryrd}
\usepackage{algorithm}
\usepackage[noend]{algpseudocode}

\usepackage{graphicx}
\usepackage[usenames,dvipsnames]{color}
\usepackage{hyperref}
\usepackage[utf8]{inputenc}
\usepackage[T1]{fontenc}
\usepackage{booktabs}
\usepackage{xspace}
\usepackage{xcolor}
\usepackage{array}
\usepackage{listings}
\usepackage{siunitx} %scientific notation

\usepackage{comment}

\newif\ifcomments%
\commentstrue%

\newcommand{\R}{\ensuremath{\mathbb{R}}}
\newcommand{\C}{\ensuremath{\mathbb{C}}}

\newcommand{\D}{{\ensuremath{\cal D}}}
\newcommand{\F}{{\ensuremath{\cal F}}}

\begin{document}

\title{Algebraic Biochemistry: a Framework for Analog Online Computation in Cells}
\author{Mathieu Hemery \and Fran\c{c}ois Fages}
\institute{Inria Saclay, Palaiseau, France}

\maketitle

\begin{abstract}
 The Turing completeness of continuous chemical reaction networks (CRNs) states that any
  computable real function can be computed by a continuous CRN on a finite set of
  molecular species, possibly restricted to elementary reactions, i.e.~with at most two
  reactants and mass action law kinetics. In this paper, we introduce a notion of
  online analog computation for the CRNs that stabilize the concentration of their output
  species to the result of some function of the concentration values of their input
  species, whatever changes are operated on the inputs during the computation. We prove
  that the set of real functions stabilized by a CRN with mass action law kinetics is
  precisely the set of real algebraic functions.
\end{abstract}
 
\section{Introduction}

Chemical Reaction Networks (CRNs) are a standard formalism used in chemistry and
biology to describe complex molecular interaction systems.
In the perspective of systems biology, they are a central tool to analyze
the high-level functions of the cell in terms of their low-level molecular interactions.
In that perspective, the Systems Biology Markup Language (SBML)~\cite{Hucka03bi}
is a common format to exchange CRN models and build CRN model repositories, such as Biomodels.net~\cite{CLN13issb}
which contains thousands of CRN models of a large variety of cell biochemical processes.
In the perspective of synthetic biology, they constitute a target
programming language to implement in chemistry new functions either \emph{in vitro},
e.g.~using DNA polymers~\cite{QSW11dna}, or in living
cells using plasmids~\cite{DWGLERBW14nar} or in artificial vesicles using proteins~\cite{CAFRM18msb}.

The mathematical theory of CRNs was introduced in the late 70's, on the one hand by Feinberg in \cite{Feinberg77crt},
by focusing on perfect adaptation properties and multistability analyses~\cite{CF06siamjam},
and on the other hand, by {\'E}rdi and T{\'o}th
by characterizing the set of Polynomial Ordinary Differential Equation systems (PODEs)
that can be defined by CRNs with mass action law kinetics,
using dual-rail encoding for negative variables~\cite{ET89book}.

More recently, a computational theory of CRNs was investigated by formally relating their
Boolean, discrete, stochastic and differential
semantics in the framework of abstract interpretation~\cite{FS08tcs},
and by studying the computational power of CRNs under those different interpretations~\cite{CSWB09ab,CDS12nc,FLBP17cmsb}.

In particular, under the continuous semantics of CRNs interpreted by ODEs,
the Turing-completeness result established in~\cite{FLBP17cmsb} states that
any computable real function, i.e.~computable by a Turing machine with an arbitrary precision given in input,
can be computed by a continuous CRN on a finite set of molecular species,
using elementary reactions with at most two reactants and mass action law kinetics.
This result uses the following notion of analog computation of a non-negative real function computed by a CRN,
where the result is given by the concentration of one species, $y_1$,
and the error is controlled by the concentration of one second species, $y_2$:

\begin{definition}\cite{FLBP17cmsb}
A function $ f: \mathbb{R_+} \to \mathbb{R_+} $ is CRN-computable if there exist
%a mass-action-law CRN $\{(R_i,P_i,f_i)\}_{i\in I}$
a CRN over some molecular species $\{y_1,...,y_n\}$,
and a polynomial $ q\in\mathbb{R_+} ^ n[\mathbb{R_+}]$ defining their initial concentration values,
such that for all $x\in\mathbb{R_+}$ there exists some (necessarily unique) function $y : \mathbb{R}\to \mathbb{R}^n$ such that
$y(0)=q(x),\ y'(t) = p (y (t)) $ and for all $t>1$
$$|y_1(t)-f(x)| \le y_2(t),$$
 $y_2(t) \ge 0$, $y_2(t)$ is decreasing and 
 $ \lim_{t \to \infty} y_2(t) = 0$.
\end{definition}

From the theoretical point of view of computability, the control of the error
which is explicitly represented in the above definition by the auxiliary variable $y_2$,
is necessary to
decide when the result is ready for a requested precision, and to mathematically define the function computed by a CRN if any.

From a practical point of view however, precision is of course an irrelevant issue
since chemical reactions are stochastic in nature and the stability of the CRN and robustness with respect to the concentration species variations
is a more important criterion
than the precision of the result.
With this provision to omit error control, the Turing-completeness result of continuous CRNs was used in~\cite{FLBP17cmsb}
to design a compilation pipeline to implement any mathematical elementary function in abstract chemistry.
This compiler, implemented in our CRN modeling, analysis and synthesis software Biocham~\cite{CFS06bi}
as the one presented here\footnote{All the computational results presented in this paper are available in an executable Biocham notebook at \url{https://lifeware.inria.fr/wiki/Main/Software\#CMSB22}},
 generates a CRN over a finite set of abstract molecular species,
through several symbolic computation steps, mainly composed of polynomialization~\cite{HFS21cmsb},
quadratization~\cite{HFS20cmsb} and lazy dual-rail encoding of negative variables.
A similar approach is undertaken in the CRN$++$ system~\cite{VSK18dna}, also related to\cite{CTT20nc}.

Now, it is worth remarking that in the definition above, and in our implementation in Biocham,
the input is defined by the initial concentration of the input species
which may be consumed by the synthesized CRN to compute the result.
This marks a fundamental difference with many natural CRNs which perform a kind of online computation
by adapting the response to the evolution of the input.
This is the case for instance of the ubiquitous MAPK signaling network which computes an input/output function
well approximated by a Hill function of order 4.9~\cite{HF96pnas},
while our synthesized CRNs to compute the same function consume their input and do not correctly adapt to change
of the input value during computation~\cite{HFS21cmsb}.

In this paper, we are interested in a notion of online computation for continuous CRNs,
by opposition to our previous notion of static computation of the result of a function for any initially given input.
Our main theorem shows that the set of input/output functions stabilized online by a CRN with mass action law kinetics,
is precisely the set of real algebraic functions.

\begin{example}\label{ex:hill}
We can illustrate this result with a simple example.
Let us consider a cell that produces a
receptor, $I$, which is transformed in an active form, $A$, when bound to an external ligand
$L$, and which stays active even after unbinding:
\begin{gather}
\begin{aligned}
  L+I &\rightarrow L+A \\
  \emptyset &\leftrightarrow I \\
  A &\rightarrow \emptyset
\end{aligned}
\end{gather}
The differential semantics with mass action law of unitary rate constant is the PODE:
\begin{gather}
\begin{aligned}
\frac{dI}{dt} &= 1-I-LI \\
\frac{dA}{dt} &= LI-A \\
\frac{dL}{dt} &= 0
\end{aligned}
\end{gather}
At steady state, all the derivatives are null and by eliminating $I$, we immediately
obtain the polynomial equation: $L - LA - A = 0$. Thinking of this simple CRN
as a kind of signal processing with the ligand as input and the active receptor
as output, it is possible to find a polynomial relation between the input and the
output. In this case, this relation entirely defines the function computed by
the CRN: $$A(L) = \frac{L}{1+L}.$$
For a given concentration of ligand, this is
the only stable state of the system, independently of the initial concentrations
of $A$ or $I$. This is why we say that the CRN stabilizes the function.
\end{example}

Such functions, for which there exists a polynomial relation between the inputs and
output, are called algebraic functions in mathematics.
We show here that the set of real algebraic functions is
precisely the set of input/output functions stabilized by CRNs with mass action law kinetics.
Furthermore, our constructive proof provides a compilation method to generate a stabilizing CRN for any
real algebraic curve, i.e. any curve defined by the zeros of some polynomial.

\subsection{Related work}

Our CRN synthesis results can be compared to the ones of Buisman \& al.~who present
in~\cite{BEHL09al} a method to implement any algebraic expression by an abstract
CRN\footnote{The terminology of ``algebraic functions'' used in the title
 of~\cite{BEHL09al} refers in fact to its restriction to algebraic expressions.}.
They
rely on a direct expression of the function and a compilation process that mimics the
composition of the elementary algebraic operations. We improve their results in three
directions. First, our compilation pipeline is able to generate stabilizing CRNs for any
algebraic function, including those algebraic functions that cannot be defined by algebraic expressions, such
as the Bring radical (see Ex.~\ref{ex:Bring}). Second, our theoretical framework shows
that the general set of algebraic functions precisely characterizes the set of functions
that can be stably computed online by a CRN.
Third, the quadratization and lazy-negative optimization algorithms presented in this paper allow us to generate more concise CRNs.
On the example given in section 3.4 of \cite{BEHL09al} 
for the quadratic expression $$y = \frac{b - \sqrt{b^2-4ac}}{2a}$$ used to find the root of a
polynomial of second order, our compiler generates a CRN of $7$ species (including the $3$ inputs) and $11$
reactions, while their CRN following the syntax of the expression uses $10$ species and $14$ reactions.
Moreover, our dual-rail encoding allows us to give correct answers for negative values of $y$.

\section{Definitions and Main Theorem}

For this article, we denote single chemical species with lower case letters and set of
species 
with
upper case letters, e.g.~$X = \{ x_1, x_2, \ldots \}$.
By abuse of notation, we will use the same symbol for the variables of the ODEs,
the chemical species and their concentrations, the context being sufficient to
remove any ambiguity.

\subsection{Chemical Reaction Networks}

A chemical reaction with mass action law kinetics is composed of a multiset of reactants, a multiset of products and one rate constant.
Such a reaction can be written as follows:
\begin{equation}
	a+b \xrightarrow{k} 2a
\end{equation}
where $k$ is the rate constant, and the multisets are represented by linear expressions in which the (stoichiometric) coefficients
give the multiplicity of the species in the multisets, here 2 for the product $a$, the coefficients equal to 1 being omitted.
In this example, the velocity of the reaction is the product
$k a b$, i.e.~the rate constant $k$ times the concentration of the reactants, $a$ and $b$.
%If the rate is equal to one, ($k=1$), we will ommit it above the arrow to avoid cluttering our notations.

In this paper, we consider CRNs with mass action law kinetics only. 
It is well known that the other kinetics, such as Michaelis-Menten or Hill kinetics,
can be derived by quasi-steady state or quasi-equilibrium reductions of more complex CRNs
with mass action kinetics~\cite{Segel84book}.
Furthermore, the Turing-completeness of this setting~\cite{FLBP17cmsb} shows that there is no loss of generality with that restriction.

We also study the case where one or several species, called \emph{pinned (input) species}, are present in such a
way that their concentrations remain constant,
independently of the activity of the CRN under study.

\begin{definition}
 The differential semantics of a CRN with a distinguished set of pinned species $S^p$,
 is composed of the usual ODEs for the non pinned species
$s \notin S^p$, and null differential functions for the pinned species:
\begin{equation}
\forall s \in S^p,\quad \frac{d s}{dt} = 0.
\end{equation}
\end{definition}

This pinning process may be due to a scale separation between the different concentrations (one
of the species is so abundant that the CRN essentially do not modify its
concentration), to a scale separation of volume (e.g. a compartment within a cell and a freely diffusive small molecule)
or to an active mechanism ensuring perfect adaptation (e.g. the
input is produced and consumed by some external reactions faster than the CRN itself, thus locking its concentration).

\subsection{Stabilization}

We are interested in the case where one particular species of the CRN, called the output, is such that, whatever moves the inputs may do, once the
inputs are fixed, the concentration of the output species stabilizes on the result of some function
of the fixed inputs. Furthermore, we want this value to be robust to small
perturbations of both the auxiliary variables and the output. Of course, if the
inputs are modified, the output has to be modified. The output thus encodes
a particular kind of robust computation of a function which we shall call
stabilization.

\begin{definition}\label{def:stabilizes}
We say that a CRN over a set of $m+1+n$ species $\{X,y,Z\}$ with 
pinned inputs $X$ of cardinality $m$ and distinguished outputs $y$, stabilizes the
function $f : I \mapsto \R_+$, with $I \subset \R_+^m$, over the domain $\D \subset \R_+^{m+1+n}$
if:
\begin{enumerate}
  \item $\forall X^0 \in I$ the restriction of the domain $\D$ to the slice
$X=X^0$ is of plain dimension $n+1$, and
	\item $\forall (X^0,y^0,Z^0) \in \D$ the Polynomial Initial Value Problem (PIVP) given by the differential semantic with
pinned input species $X$ and the initial conditions $X^0,y^0,Z^0$ is such that:
$\lim_{t\rightarrow \infty} y(t) = f(X).$
\end{enumerate}

This definition is extended to functions of $\R^n$ in $\R$ by dual-rail
encoding~\cite{ET89book}: for a CRN over the species $\{X^+, X^-, y^+, y^-, Z\}$ we ask that
$\lim_{t\rightarrow \infty} (y^+ - y^-)(t) = f(X^+ - X^-),$ for all initial
conditions in the validity domain $\D$.

Let $\F_S$ be the set of functions stabilized by a CRN.
\end{definition}

Several remarks are in order.
A first remark concerns the fact that we ask for a domain $\D$ of plain
dimension $n+1$, i.e.~non-null measure in $\mathbb{R}^{n+1}$. That constraint
is imposed in order to benefit from a strong form of robustness: there exists an
open volume containing the desired fixed point such that it is the unique
attractor in this space. Hence in this setting, minor perturbations are always
corrected. This requirement of an isolated fixed point also impedes 
from hiding information in the initial conditions. The following example
illustrates the crucial importance of that condition on the dimension of the
domain $\D$

\begin{example}\label{ex:cosine}
The following PODE is constructed in such a way that $z_2$ goes exponentially to $x$ while $y$ and $z_1$ remain
equal to $\cos(z_2)$ and $\sin(z_2)$ respectively.
\begin{gather}
\begin{aligned} 
  \frac{dx}{dt} &= 0, \quad &x(t=0) = \text{input} \\
  \frac{dy}{dt} &= -z_1 (x-z_2) \quad &y(t=0) = 1 \\
  \frac{dz_1}{dt} &= y (x-z_2) \quad &z_1(t=0) = 0 \\
  \frac{dz_2}{dt} &= (x-z_2) \quad &z_2(t=0) = 0\\
\end{aligned}
\end{gather}
One might think that this
PODE stabilizes the cosine as we have $\lim y(t) = cos(x)$ for any value of $x$. 
But cosine is not an algebraic function, and indeed, the only requirement for
this PODE to be at steady state is: $x = z_2$, meaning that there exist fixed
points for any value of $z_1$ and $x$. So this PODE does not stabilize the cosine
function.
The reason is that the cosine computation is encoded in the initial state. It is only for the
domain where $y = \cos(z_2)$ and $z_1 = \sin(z_2)$ that the computation works, but this domain
is of null measure in $\R^3$ which breaks the first condition of Def.~\ref{def:stabilizes}.
\end{example}

A second remark is that
since the inputs are fixed in our semantics (they are by definition pinned
species),
the target of the output species which is the result $f(X)$ of some function $f$
is not a fluctuating goal: it is fixed by the initial conditions.
In practice, what we ask is that the dynamics of the ODE for the slice of the
domain $\D$ defined by imposing the inputs have a unique attractor satisfying $y
= f(X^0)$. But as we do not impose any constraint on the other variables ($Z$),
we cannot speak of a fixed point since the dynamics on the other variables may
not stabilize (e.g.~oscillations, divergence, etc.). We will nevertheless speak
of these object as pseudo-fixed point. If we start from a point on this pseudo
fixed point, we will have: $\forall t, y(t) = f(X^0)$.

A third remark is that our definition implies that apart from a transient
behaviour of characteristic time $\tau$, the whole system is constrained to live
in, or nearby, the subspace defined by $y = f(X)$. What is interesting is that
if the inputs are themselves varying with a characteristic time that is slower
than $\tau$, the output will follow those variations, hence preserving our
desired property up to an error coming from the delay as long as the system
stays in the domain $\D$. In a synthetic biology perspective, it is in
principle possible to use a time-rescaling to modify the value of $\tau$.
While a small $\tau$ allows for a faster adaptation, this usually comes at the expense of a
greater energetic cost as the proteins turn-over tends to
increase.

\subsection{Algebraic Curves and Algebraic Functions}

In mathematics, an algebraic curve (or surface) is the zero set of a polynomial
in two (or more) variables. It is a usual convention in mathematics to
speak indifferently of the polynomial and the curve it defines, seen as the same object.
For instance, $x^2+y^2-1$ is seen as the unit circle.

Any polynomial $P$ can be expressed as a product of irreducible polynomials,
i.e.~polynomials that cannot be further factorized, up to a 
constant $k$:
$$\displaystyle{P = k. \prod_{i = 1 \dots n} P_i^{a_i}}.$$
The $P_i$'s are called the components of $P$, and $a_i$ the multiplicity of
$P_i$.
We say that $P$ is in reduced form if all the components have multiplicity one, $\forall i\ a_i=1$.
This is justified by one important result of
algebraic geometry: in an algebraic closed field, such as the complex numbers $\C$, the set of
points of an algebraic curve given with their multiplicity, suffices to define the polynomial in reduced form.
This makes algebraic geometry an elegant and powerful theory. 

In a non-algebraically closed field such as $\R$, a polynomial may have no real
root. This difficulty is however irrelevant to us in this paper since we start
with an algebraic real function, thus assuming the existence of real roots. For
the purpose of this article, this fundamental result provides a canonical
correspondence between an algebraic real function and its polynomial of minimal
degree, i.e.~a polynomial in reduced form, up to a multiplicative factor.

\begin{definition}
A function $f:I \subset \R^m \mapsto \R$ is algebraic if there exists
a polynomial $P$ of $m+1$ variables such that:
\begin{equation}
\forall X \in I, P(X, f(X)) = 0.
\end{equation}
We denote $\F_A$ the set of real algebraic functions.
\end{definition}

We shall prove the following central theorem:
\begin{theorem}\label{thm:main}
 The set of functions stabilized by a CRN with mass action law kinetics
 is the set of algebraic real functions:
 $\F_S = \F_A.$
\end{theorem}

One technical difficulty comes from the fact that it is not immediate to determine the
function $f$ from the polynomial. Indeed for a given polynomial pinning the value of the
inputs results in one, several or no possible value for the output. Hence, a given
polynomial actually defines several functions on the domain of its input.
This is for instance the case of the unit circle curve defined by $x^2+y^2-1$.
If we see it as an equation to solve upon $y$, it admits
two solutions when $x \in ]-1,1[$, exactly one for $x=-1$ or $x=1$,
and no solution for other values of $x$. Hence, that curve defines two continuous
functions $y(x)$, each of them with support $]-1,1[$.

To overcome that difficulty, let us call branch point (or branch curve), the set of points
where the number of real roots of an algebraic function changes (${-1,1}$ in the previous
example). Now for a polynomial $P(X,y)$ and a given root $X,y$ that is not a branch point,
the implicit function theorem ensures the existence and uniqueness of the implicit function
up to the next branch point/curve.

\begin{example}\label{ex:unitcircle}
 The branch points of the unit circle polynomial $x^2+y^2-1$ are ${(-1,0),(1,0)}$.
 If we provide an additional point on the curve, e.g.~$(0,1)$,
 one can define the function that contains it and that goes from one branch point to the
other one, here:
\begin{align*}
  ]-1,1[ &\rightarrow \R \\
  f: x &\mapsto \sqrt{x^2-1}
\end{align*}
Fig.~\ref{fig:circle} in a latter section illustrates the flow diagram used in this example by our CRN compiler to approximate that function.

Similarly in the case of the sphere defined by the polynomial $x_1^2+x_2^2+y^2-1$,
the branch curve is the whole unit circle contained in the plane $y=0$. And
giving the point $0,0,-1$ is enough to define the whole surface corresponding to
the down part of the sphere inside the branch-curve circle.
\end{example}

\section{Proof}

\begin{lemma}\label{lmm:<-}
$\F_A \subset \F_S.$
\end{lemma}

\begin{proof}
Suppose that $f:I \mapsto \R$ is an algebraic real function and let $P_f$ denote the canonical polynomial
such that $\forall X \in I, P_f(X,f(X)) = 0$. Let us choose a vector
$X_0$ in the domain of $f$.

Then, the PODE
\begin{gather}
\begin{aligned}
  \frac{dy}{dt} &= \pm P(X,y), \label{eq:stab} \\
  \frac{dX}{dt} &= 0,
\end{aligned}
\end{gather}
is such that $Y = f(X)$ is a fixed point. By choosing the sign such that,
locally $\pm P(X,y)$ is negative above $Y = f(X)$ and positive below, we ensure
that this point is stable.

The fact that the polynomial has to change the sign across the fixed point is
dut to the fact that we choose the polynomial of minimal degree, hence it
has to be in reduced form and the multiplicity of every branch of the curve is
one: the sign cannot be the same on both sides of the curve.

	It is worth remarking that any ODE system made of elementary mathematical
	functions can be transformed in a polynomial ODE system~\cite{HFS21cmsb}, hence one can wonder why we restrict here
	to polynomial expressions. This comes from the condition that asks that the
	domain $\D$ be of plain dimension in Def.~\ref{def:stabilizes}.
    The polynomialization of an ODE system may indeed introduce constraints
	between the initial concentrations which is precisely what is forbidden by the
	requirement upon $\D$.

Now, let us note $Y^+ = \inf(Y\ |\ P(X,Y)=0, Y>f(X))$ and $Y^- = \sup(Y\ |\ P(X,Y)=0, Y<f(X))$,
with $\pm \infty$ values if the set is empty.
We know that for all $y$
in $]Y^-, Y^+[$, the only attractor is $f(X)$ and as a polynomial can only have
a finite number of zeros, those sets are non empty.

For all variables that are not bound to be
positive, the dual-rail encoding consists in splitting the variable into two positive variables corresponding to the positive and
negative parts. Then, all positive monomials can be dispatched to the positive part and all
negative ones to the negative part (with a positive sign), with the addition of an mutual annihilation reaction
between the variables as described in~\cite{FLBP17cmsb}.
It is worth noting that that dual-rail encoding is necessary for positive
functions whenever the auxiliary variables may take negative values.

\end{proof}

\begin{lemma}\label{lmm:->}
$\F_S \subset \F_A.$
\end{lemma}

\begin{proof}
Let us suppose that $f$ is a function stabilized by a mass action CRN.
The idea is to use the characterization of functions that are projectively
polynomial, as defined in~\cite{CPSW05ejde}. By using the higher-order derivatives of
the stabilized variable, it is shown in~\cite{CPSW05ejde} that one can eliminate all the auxiliary
variables and obtain a single equation of the form:
$$P(X, y, y^{(1)}, \ldots, y^{(n)}) = 0.$$

Using the fact that for all $X$, $y=f(X)$ is a pseudo fixed point by definition, if we use it
	as initial condition we immediately get:
\begin{align*}
  X &= X, \\
  y &= f(X), \\
	y^{(k)} &= 0 \quad \forall k \in [1,n].
\end{align*}

Injecting this in the characterization of the function $y$, we obtain:
\begin{equation}
  \forall X, P^\star(X, f(X)) = 0.
\end{equation}

There are now two cases. Either $P^\star$ is not trivial and effectively
defines the surface of fixed points: this gives a polynomial for $f$, hence $f$
is algebraic. Either $P^\star$ is the uniformly null polynomial. But in this
case, every points in the $X,y$ plane may be a fixed point and the domain $\D$
of the definition of stabilization is reduced to a single point, yet we asked it
to be of non-null measure. Therefore, $P^\star$ is not trivial and $f$ is algebraic.
\end{proof}

\section{Compilation Pipeline for Generating Stabilizing CRNs} \label{seq:compilation}

The proof of lemma~\ref{lmm:<-} is constructive and provides a method 
to transform any algebraic function defined by a polynomial and one point,
in an abstract CRN that stabilizes it.
This is implemented with a
command

\verb$stabilize_expression(Expression, Output, Point)$\\
with three arguments:
\begin{description}
\item[{\tt Expression}:] For a more user friendly interface, we accept in input more general mathematical
expressions than polynomials; the non polynomial parts are detected
and transformed by introducing new variable/species to compute their values;
\item[{\tt Output}:] a name of the Output species different from the input;
\item[{\tt Point}:] a point on the algebraic curve that is used to
determine the branch of the curve to stabilize if several exist.
\end{description}

Similarly to our previous pipeline for compiling any elementary function in an
abstract CRN that computes it~\cite{HFS21cmsb,HFS20cmsb,FLBP17cmsb}, 
our compilation pipeline for generating stabilizing CRNs follows the same sequence of symbolic transformations: 
\begin{enumerate}
\item polynomialization
\item stabilization
\item quadratization
\item dual-rail encoding
\item CRN generation
\end{enumerate}
yet with some important differences.

\subsection{Polynomialization}

This optional step has been added just to obtain a more
user friendly interface, since polynomials may sometimes be cumbersome to manipulate.
The first argument thus admits algebraic expressions instead of
being limited to polynomials.

The rewriting simply consists in detecting all the non-polynomial terms of the form
$\sqrt[a]{b}$ or $\frac{a}{b}$ in the initial
expression and replace them by new variables, hence obtaining a polynomial.

Then to compute the variables that just have been introduced, we perform
the following basic operations on each of them to recover polynomiality:
\begin{align*}
n = \sqrt[a]{b} &\rightarrow n^a-b \\
n = \frac{a}{b} &\rightarrow nb-a
\end{align*}
and recursively call \verb|stabilize_expression| on these new expressions
with the introduced variable (here $n$) as desired output.

\subsection{Stabilization}
To select the branch of the curve to stabilize,
it is sufficient to choose the sign in front of the polynomial in equation~\ref{eq:stab}.
such that at the designated point, the second derivative
of $y$ is positive. For this, we use a formal derivation to compute the sign of
the polynomial, and reverse it if necessary.

\subsection{Quadratization}

The quadratization of the PODE is an optional transformation which aims at generating
elementary reactions, i.e.~reactions having at most two reactants each,
that are fully decomposed and more amenable to concrete implementations with real biochemical molecular species.
It is worth noting that the quadratization problem to solve
here is a bit different from the one of our original pipeline studied
in~\cite{HFS20cmsb} since we want to preserve a different property of the PODE.
It is necessary here that the introduced variables stabilize on their target value
independently of their initial concentrations.
While it was possible in our previous framework to initialize the different
species with a precise value given by a polynomial of the input, this is no more the case here as
the domain $\D$ has to be of plain dimension. 

The variables introduced by quadratization correspond to monomials of order higher than $2$ that can
thus be separated as the product of two variables corresponding to monomials of
lower order: $A$ and $B$. Those variables can be either present in the original
polynomial or introduced variables.
The following set of reactions:
$$A+B \rightarrow A+B+M$$
$$M \rightarrow \emptyset,$$
gives the associated ODE:
\begin{equation}\frac{dM}{dt} = AB-M, \label{eq:monomial_computation}\end{equation}
for which the only stable point satisfies: $M = AB$.

Furthermore as before, we are interested in computing a quadratic PODE of minimal dimension.
In~\cite{HFS20cmsb}, we gave an algorithm in which the introduced variables were always equal to the monomial they compute,
whereas in our online stabilization setting, this is true only when $t \rightarrow \infty$.
For instance, if we replaced $AB$ by $M$ in equation \ref{eq:monomial_computation}, the system would no longer adapt to changes of the input.
To circumvent this difficulty, it is 
possible to modify the PIVP and use it as input of our previous algorithm to
take this constraint and still obtain the minimal set of variables.
In our previous computation
setting, the derivatives of the different variables where simply the
derivatives of the associated monomials computed in the flow generated by the
initial ODE.
In Alg.~\ref{algo:hack}, we construct a pseudo-ODE containing twice as many
variables, the derivatives of which being built to ensure
that the solution is correct. The idea is that the actual variables are of the
form $Mb$ and the $Mb^2$ variables exist only to construct the solution. To compute
quadratic monomials with a
$b^2$ term present in the derivatives of the $Mb$ variables (the
``true'' variables),
one can either add two $Mb$ variables to the solution set or add
a single $Mb^2$ variable. As can be seen on the lines 5 and 9 of Alg.~\ref{algo:hack}, $Mb^2$ variables
require that the corresponding $Mb$ is in the solution set.

\begin{algorithm}
\caption{Quadratization algorithm for a PODE stabilizing a function.
The $minimal\_quadratic\_set(PODE, y)$ returns the
minimal set of variables containing $y$ sufficient to express all its derivatives 
in quadratic form~\cite{HFS20cmsb}.\label{algo:hack}}
\textbf{Input}: A PODE of the form $\frac{dx_i}{dt} = P_i(X)$, with $i \in [1,n]$ to compute $x_n$\\
\textbf{Output}: A set $S$ of monomials to quadratize the input.
\begin{algorithmic}[5]
\State $ODE_\text{aux} \gets \emptyset$
\State find an unused variable name: $b$
\ForAll{$i \in [1,n]$}
  \State add $\frac{dx_ib}{dt} = P_i(X) \times b^2$ to $ODE_\text{aux}$
  \State add $\frac{dx_ib^2}{dt} = x_ib$ to $ODE_\text{aux}$
\EndFor
\State $AllMonomials \gets$ the set of monomials that are less or equal to a
monomial present in one of the $P_i$ and not in $X$.
\ForAll{$M \in AllMonomials$}
  \State add $\frac{dMb}{dt} = Mb^2$ to $ODE_\text{aux}$
  \State add $\frac{dMb^2}{dt} = Mb$ to $ODE_\text{aux}$
\EndFor
\State $S_\text{aux} \gets minimal\_quadratic\_set(ODE_\text{aux}, x_n b)$
\State $S \gets \emptyset$
\ForAll{$Mb \in S_\text{aux}$}
  \State add $M$ to $S$
\EndFor
\State \textbf{return} $S$
\end{algorithmic}
\end{algorithm}

This variant of the quadratization problem studied in~\cite{HFS20cmsb} has the same
theoretical complexity, as shown by the following proposition:

\begin{proposition}
The quadratization problem of a PODE for stabilizing a function and minimizing the number of variables is NP-hard.
\end{proposition}
\begin{proof}
The proof proceeds by reduction of the vertex covering of a graph as in~\cite{HFS20cmsb}.
Let us consider the graph $G=(V,E)$ with vertex set $v_i, i \in [1,n]$ and edge set $E \in V\times V$.
And let us study the quadratization of the PODE with input variables $V \cup
\{a\}$ and output variable $y$ such that the $y$ computes
$\sum_{v_i v_j \in E} v_i v_j a$.
The derivative is:
\begin{equation}
\frac{dy}{dt} = \sum_{v_i v_j \in E} v_i v_j a - y.
\end{equation}
An optimal quadratization contains variable corresponding either to $v_i a$ or
$v_i v_j$ indicating that an optimal covering of the graph $G$ contains either
the node $v_i$ either indifferently $v_i$ or $v_j$. Hence en optimal
quadratization gives us an optimal covering which concludes the proof.
\end{proof}

Our previous MAXSAT algorithm \cite{HFS20cmsb} and heuristics \cite{HFS21cmsb} can again be used here with the slight modification
mentioned above concerning the introduction of new variables. 

Alg.~\ref{algo:hack}, when invoked with an optimal search for
$minimal\_quadratic\_set$, is nevertheless not guaranteed to generate optimal solutions,
because of the ``pseudo'' variables noted $Mb^2$.
Despite those theoretical limitations, the CRNs generated by Alg.~\ref{algo:hack} are particularly concise,
as shown in the example section below and already mentioned above for the compilation of algebraic expressions compared to \cite{BEHL09al},

\subsection{Lazy dual-rail encoding}
As in our original compilation pipeline~\cite{FLBP17cmsb},
it is necessary to modify our PODE in order to impose that no variable may
become negative. This is possible through a lazy version of dual-rail encoding.
First by detecting the variable that are or may become negative
and then by splitting them between a positive and negative part, thus
implementing a dual-rail encoding of the variable: $y = y^+ - y^-$.
Positive terms of the
original derivative are associated to the derivative of $y^+$ and negative terms
to the one of $y^-$ and a fast mutual degradation term is finally associated to both
derivative in order to impose that one of them stays close to zero~\cite{FLBP17cmsb}.

\subsection{CRN generation}
The same back-end compiler as in our original pipeline
is used,
i.e.~introducing one reaction for each monomial.
It is worth remarking that this may have for effect to aggregate in one reaction
several occurrences of a same monomial in the ODE system~\cite{FGS15tcs}.

\section{Examples}

\begin{example}\label{ex:circle}
As a first example, we can study the unit circle: $x^2+y^2-1$.
Our pipeline gives us for the upper part of the circle, the following CRN.
\begin{gather}
\begin{aligned}
  \emptyset &\rightarrow y^+ &\quad
2 y^- &\rightarrow 3 y^- \\
  2 x &\rightarrow y^- + 2 x &\quad
2 y^+ &\rightarrow y^- + 2 y^+ \\
y^+ + y^- &\xrightarrow{\text{fast}} \emptyset
\end{aligned}
\end{gather}
the flow of the PODE associated to this model can be seen in figure~\ref{fig:circle}\textbf{A} and the
steady state is depicted in figure~\ref{fig:circle}\textbf{B} as a function of
$x$ in the positive quadrant.
\end{example}

\begin{figure}
\centering
\textbf{A.}\includegraphics[width=0.4\textwidth]{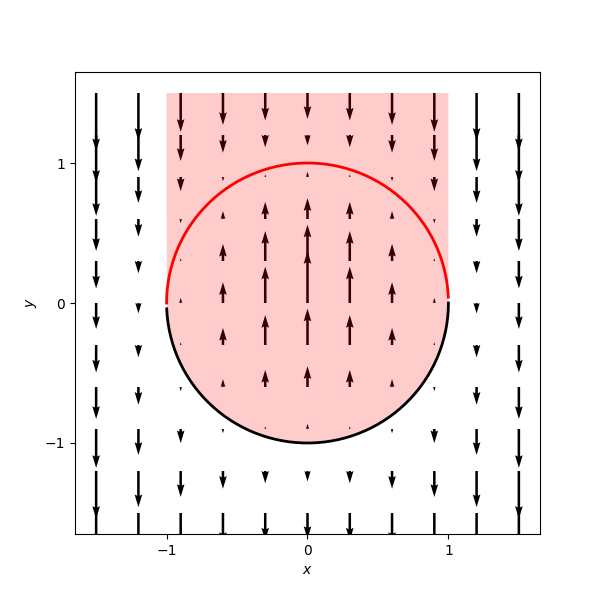}
\textbf{B.}\includegraphics[width=0.5\textwidth]{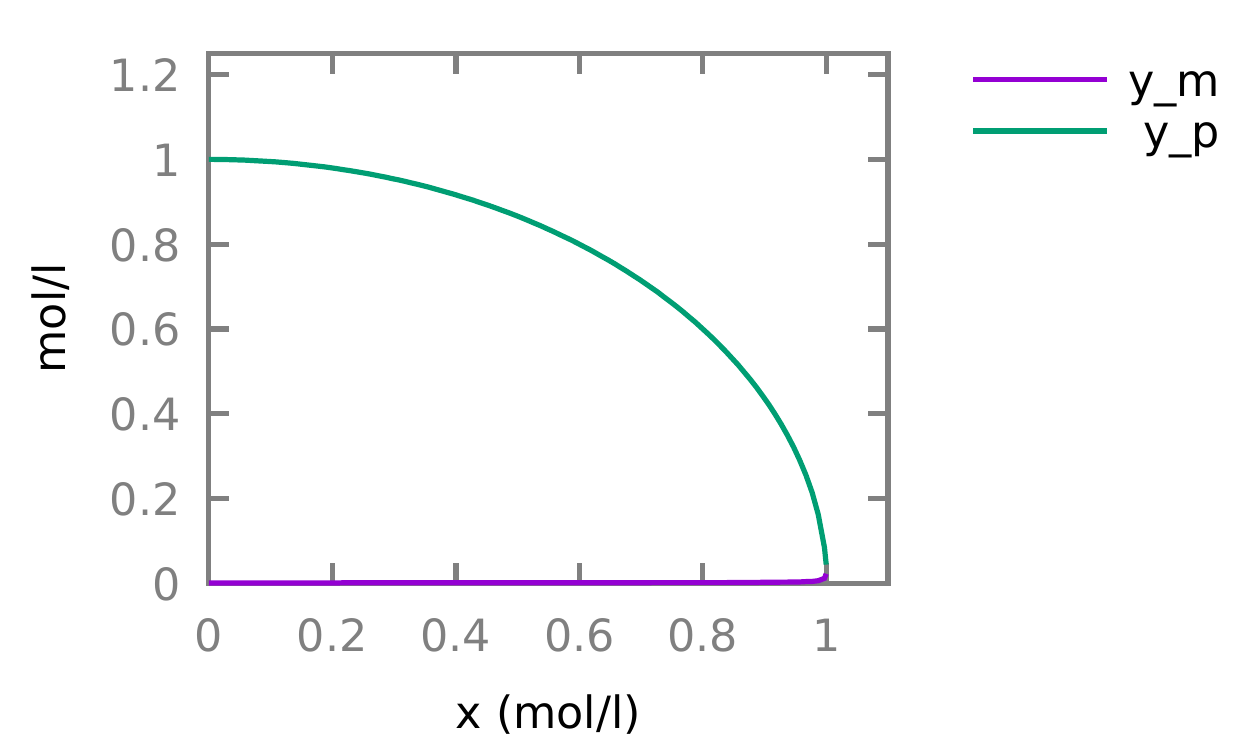} \quad
\caption{
\label{fig:circle}
\textbf{A.} Flow diagram in the $x,y$ plane before the dual rail encoding for
the stabilization of the unit circle. The arrows indicate the direction and strength
of the flow. The upper part of the curve (in red) indicates the stable branch
of the system and we colored in light red the domain $\D$ in which the system
will reach the desired steady state. Outside of $\D$, the system is driven to the
divergent state $\lim y = -\infty$.
\textbf{B.} Dose response diagram of the generated CRN where the input
concentration ($x$) is gradually increased from $0$ to $1$ while recording the steady
state value of the output species $y^+, y^-$.
}
\end{figure}

\begin{example}\label{ex:serpentine}
Even rather involved algebraic curves need surprisingly few species and
reactions.
This is the case of the serpentine curve, or anguinea,
defined by the polynomial
$(y-2) \left((x-10)^2+1\right) = 4 (x-10)$ for which we choose the point $x=10, y=2$
to enforce stability.
	The compilation process takes less than 100ms on a typical laptop\footnote{An Ubuntu 20.04, with
	an Intel Core i6, $2.4$GHz x $4$ cores and $15.5$GB of memory.}.
The generated CRN reproduces the anguinea curve on the $y$ variable, as shown in Fig.~~\ref{fig:anguinea}.
It is composed of the following $4$ species and
$12$ reactions:

\begin{gather}
\begin{aligned}
\label{model:anguinea}
  y_m+y_p &\xrightarrow{\text{fast}} \emptyset, &\quad
 2 x &\rightarrow z+2 x,\\
  z &\rightarrow \emptyset, &\quad
 \emptyset &\xrightarrow{162} y_p,\\
  y_m &\xrightarrow{101} \emptyset, &\quad
 x+y_p &\xrightarrow{20} x+2 y_p,\\
  z+y_m &\rightarrow z, &\quad
 z &\xrightarrow{2} z+y_p,\\
  x &\xrightarrow{36} x+y_m, &\quad
 y_p &\xrightarrow{101} \emptyset,\\
  x+y_m &\xrightarrow{20} x+2 y_m, &\quad
 z+y_p &\rightarrow z.
\end{aligned}
\end{gather}
\end{example}

\begin{figure}
\centering
\includegraphics[width=0.7\textwidth]{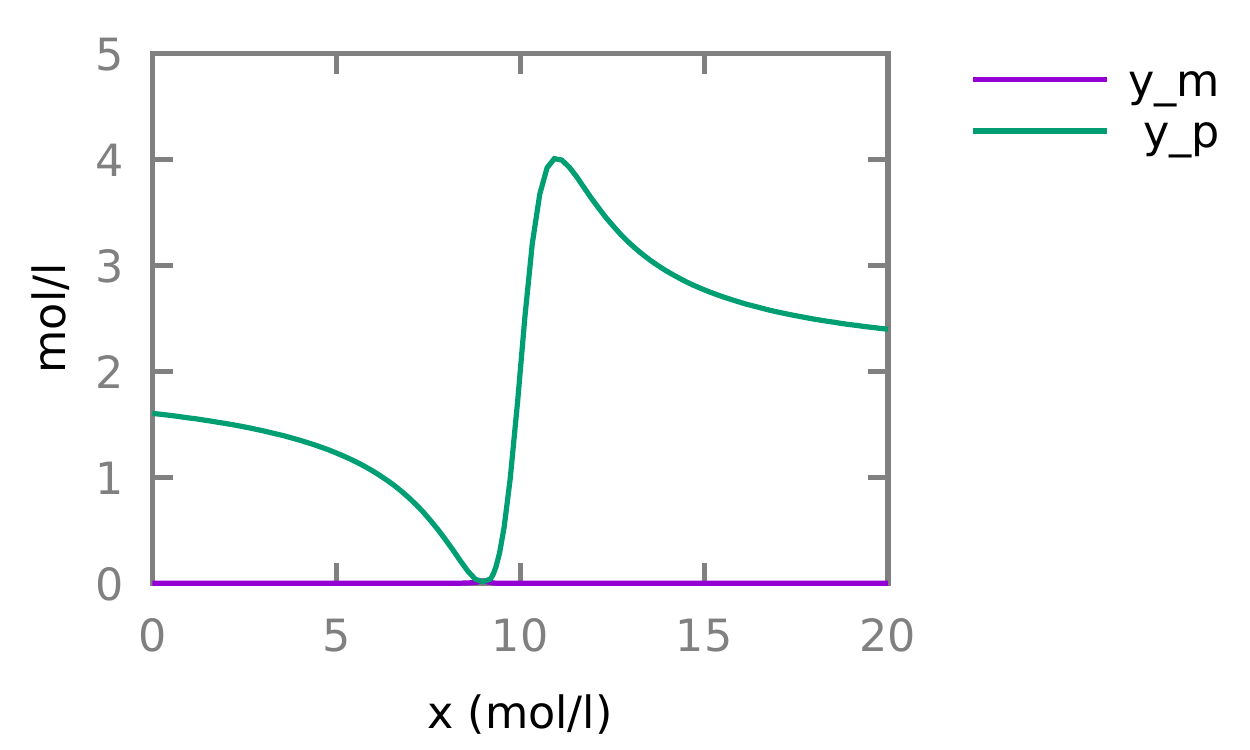}
\caption{\label{fig:anguinea}
 Dose-response diagram of the CRN generated by compilation of the serpentine algebraic curve,
 $(y-2) \left((x-10)^2+1\right) = 4 (x-10)$,
 with $x$ as input and $y$ as output.
}
\end{figure}

\begin{example}\label{ex:Bring}
In the field of real analysis, the Bring radical of a real number $x$ is defined
as the unique real root of the polynomial: $y^5+y+x$. The Bring radical is an
algebraic function of $x$ that cannot be expressed by any algebraic expression.

The stabilizing CRN generated by our compilation pipeline is composed of
  $7$ species (${y_m, y_p, y2_m, y2_p, y3_m, y3_p, x}$) and $20$ reactions presented in
  model~\ref{model:Bring}.
A dose-response diagram of that CRN
is shown in Fig.~\ref{fig:bring}.

\begin{gather}
\begin{aligned}
\label{model:Bring}
  y_m+y_p &\xrightarrow{\text{fast}} \emptyset, &\quad
  y2_m+y2_p &\xrightarrow{\text{fast}} \emptyset, \\
  y3_m+y3_p &\xrightarrow{\text{fast}} \emptyset, &\quad
  y_p &\xrightarrow{} \emptyset, \\
  y2_m+y3_p &\xrightarrow{} y2_m+y3_p+y_p, &\quad
  y2_p+y3_m &\xrightarrow{} y2_p+y3_m+y_p, \\
  x &\xrightarrow{} x+y_m, &\quad
  y_m &\xrightarrow{} \emptyset, \\
  y2_p+y3_p &\xrightarrow{} y2_p+y3_p+y_m, &\quad
  y2_m+y3_m &\xrightarrow{} y2_m+y3_m+y_m, \\
  2 \cdot y_p &\xrightarrow{} y2_p+2 \cdot y_p, &\quad
  2 \cdot y_m &\xrightarrow{} y2_p+2 \cdot y_m, \\
  y2_p &\xrightarrow{} \emptyset, &\quad
  y2_m &\xrightarrow{} \emptyset, \\
  y2_p+y_p &\xrightarrow{} y2_p+y3_p+y_p, &\quad
  y2_m+y_m &\xrightarrow{} y2_m+y3_p+y_m, \\
  y3_p &\xrightarrow{} \emptyset, &\quad
  y2_p+y_m &\xrightarrow{} y2_p+y3_m+y_m, \\
  y2_m+y_p &\xrightarrow{} y2_m+y3_m+y_p, &\quad
  y3_m &\xrightarrow{} \emptyset.
\end{aligned}
\end{gather}
\end{example}
   
\begin{figure}
\centering
\includegraphics[width=0.7\textwidth]{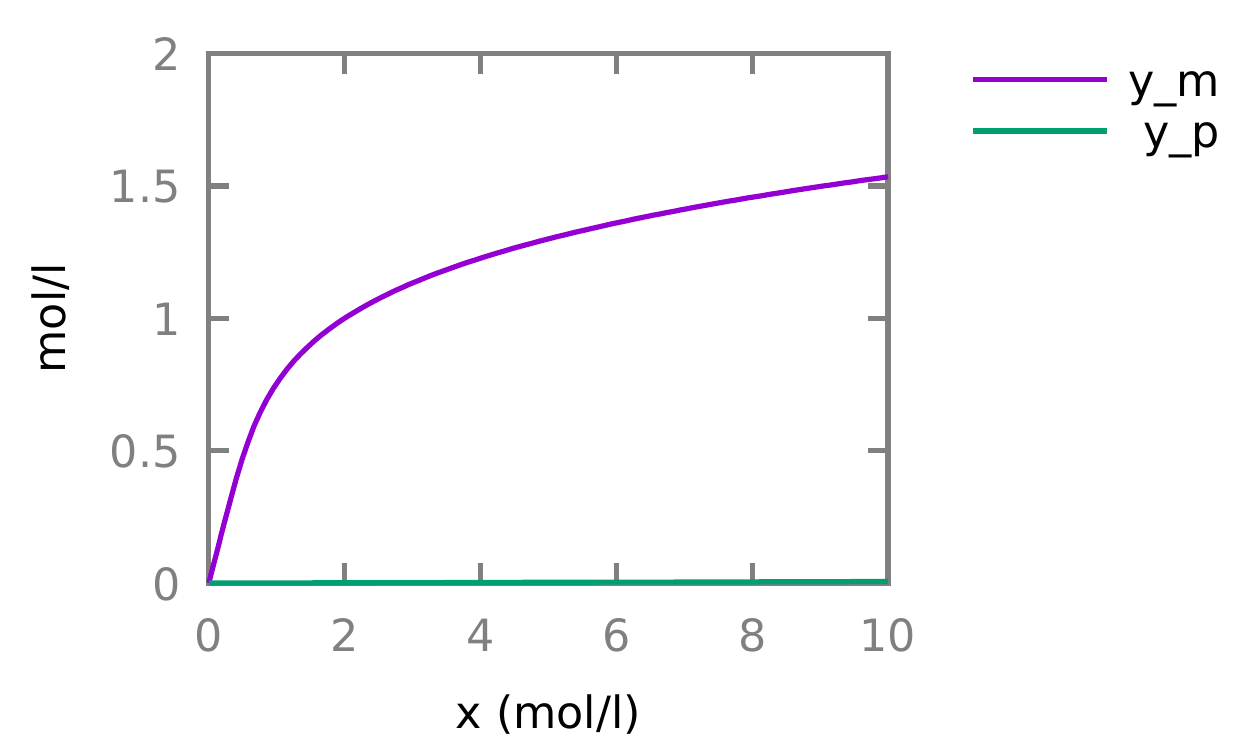}\\
\caption{
\label{fig:bring}
The bring radical is the real root of the polynomial equation $y^5+y+x=0$. As
this quantity is negative, the output is read on the negative part of the output
$y_m$. It is the
simplest equation for which there is no expression for $y$ as a function of $x$.
}
\end{figure}

\begin{example}\label{ex:hill5}
To generate the CRN that stabilize the Hill function of order 5, we can use the
expression $y-\frac{x^5}{1+x^5}$ along with the point $x=1, y=\frac{1}{2}$.
Our compilation pipeline generates the following model with $6$ species and $10$
	reactions:
\begin{tabular}{lllll}\label{model:hill5}
 $2 x \rightarrow z_1+2 x$, & $2 z_1 \rightarrow z_2+2 z_1$, & $x+z_2 \rightarrow x+z_2+z_3$, & $\emptyset \rightarrow z_4$, & $z_4+z_3 \rightarrow z_3+y$, \\
 $z_1 \rightarrow \emptyset$, & $z_2 \rightarrow \emptyset$, & $z_3 \rightarrow
	\emptyset$, & $z_4 \rightarrow \emptyset$, & $y \rightarrow \emptyset$,
\end{tabular}\\
all kinetics being mass action law with unit rate.
The $z's$ are auxiliary variables corresponding to the following expressions:
$$z_1 = x^2, \quad
z_2 = x^4, \quad
z_3 = x^5, \quad
z_4 = \frac{1}{1+x^5}.
$$

The production and degradation of $z_4$ may be surprising, but looking at all
the reactions implying both $z_4$ and $y$, we can see that their sum follow the
equation $\frac{d(z_4+y)}{dt} = 1-(z_4+y)$ hence ensuring that the sum of the
two is fixed independently of their initial concentrations.
It is worth remarking that another way of
reaching the same result would be to directly build-in conservation laws into
our CRN, hence using both steady state and invariant laws to define our steady
state which however would make us sensitive to the initial concentrations.
\end{example}

\section{Conclusion and Perspectives}

We have introduced a notion of on-line analog computation for CRNs in which the
concentration of one output species stabilizes to the result of some function of
the concentrations of the input species, whatever perturbations are applied to
the species concentrations during computation before the inputs stabilize. We
have shown that the real functions that can be stably computed by a CRN in that
way is precisely the set of real algebraic functions, defined by a
polynomial and one point. Furthermore, we have derived from the constructive
proof of this result a compilation pipeline to transform any algebraic function
in a stabilizing CRN which computes it online.

These results open a whole research avenue for both the understanding of the structure of
natural CRNs that allow cells to adapt to their environment, and for the design of
artificial CRNs to implement high-level functions in chemistry. In the latter perspective
of synthetic biology, our compilation pipeline makes it possible to automatically generate
an abstract CRN which remains to be implemented with real enzymes, as
in~\cite{CAFRM18msb}. Taking into account a catalog of concrete enzymatic reactions
earlier on in our compilation pipeline, in the polynomialization, quadratization and
dual-rail encoding steps, is a particularly interesting, yet hard, challenge in order to
guide search towards both concrete and economical solutions.

Our main theorem describes only CRNs at steady state only, while important aspects of signal
processing are linked to the temporal evolution of the signals.
Since our computation framework relies on ratios between production and degradation, a multiplication of
both terms by some factor might be the matter of future work to control the characteristic time $\tau$ of equilibration, with high
value of $\tau$ filtering out the high frequency noise of the inputs, and small $\tau$ values
resulting in a more accurate output, yet at the expense of a higher protein turnover.

\subsubsection*{Acknowledgments.} We are grateful to Amaury Pouly and Sylvain Soliman for interesting discussions on this work,
and to ANR-20-CE48-0002 and Inria AEx GRAM grants for partial support.

\bibliographystyle{plain}
\bibliography{contraintes}

\end{document}